\documentclass{llncs}
\usepackage{times}
\usepackage{latexsym}
\usepackage{color}
\usepackage{graphicx}
\usepackage{pgf}
\usepackage{tikz}
\usepackage{mathrsfs}
\usepackage{amsmath}
\usepackage{epsfig}
\usepackage{url}
\usepackage{amssymb}
\usepackage{xspace}
\usepackage{cite}
\usepackage[vlined,ruled,linesnumbered]{algorithm2e}
\usepackage{pifont}
\usepackage{subfigure}

\newcommand{\myboldmath}{}
\newcommand{\defn}[1]{{\textit{\textbf{\myboldmath #1}}}}
\newcommand{\parhead}[1]{\noindent{\textbf{#1.}\xspace}}

\newcommand{\arrival}{t}
\newcommand{\departure}{d}
\newcommand{\window}{w}
\newcommand{\client}{c}
\newcommand{\channel}{R}
\newcommand{\schedule}{S}
\newcommand{\CH}{\mathscr{R}\xspace}
\newcommand{\CL}{\mathscr{C}\xspace}

\newcommand{\preemptive}{Preemptive Reallocation\xspace}
\newcommand{\lazy}{Lazy Reallocation\xspace}
\newcommand{\constant}{Classified Reallocation\xspace}

\begin{document}

\title{
Dynamic Windows Scheduling with Reallocation~\thanks{This work was supported in part by the National Science Foundation (CCF- 0937829, CCF-1114930) and Kean University UFRI grant.}
}

\titlerunning{Dynamic Windows Scheduling with Reallocation}  
%
\author{
   Mart\'{\i}n~Farach-Colton\inst{1}
   \and
   Katia~Leal\inst{2}
   \and
   Miguel~A.~Mosteiro\inst{3}
   \and
   Christopher Thraves\inst{2}
}

\authorrunning{M.~Farach-Colton,~K.~Leal,~M.~A.~Mosteiro,~C.~Thraves} 
\institute{
Dept. of Computer Science, Rutgers University, Piscataway, NJ, USA \& Tokutek Inc.,\\
\email{farach@cs.rutgers.edu}
\and
GSyC, Universidad Rey Juan Carlos, Madrid, Spain,\\
\email{\{katia,cbthraves\}@gsyc.es}
\and
Dept. of Computer Science, Kean University, Union, NJ, USA,\\
\email{mmosteir@kean.edu}
} 

\maketitle

\begin{abstract}

We consider the Windows Scheduling problem. The problem is a restricted version of Unit-Fractions Bin Packing, and it is also called Inventory Replenishment in the context of Supply Chain. In brief, the problem is to schedule the use of communication channels to clients. Each client $\client_i$ is characterized by an active cycle  and a window $\window_i$. During the period of time that any given client $\client_i$ is active, there must be at least one transmission from $\client_i$ scheduled in any $\window_i$ consecutive time slots, but at most one transmission can be carried out in each channel per time slot. The goal is to minimize the number of channels used. 

We extend previous online models, where decisions are permanent, assuming that clients may be reallocated at some cost. We assume that such cost is a constant amount paid per reallocation. That is, we aim to minimize also the number of reallocations. We present three online reallocation algorithms for Windows Scheduling. We evaluate experimentally these protocols showing that, in practice, all three achieve constant amortized reallocations with close to optimal channel usage. Our simulations also expose interesting trade-offs between reallocations and channel usage. 
We introduce a new objective function for WS with reallocations, that can be also applied to models where reallocations are not possible. We analyze this metric for one of the algorithms which, to the best of our knowledge, is the first online WS protocol with theoretical guarantees that applies to scenarios where clients may leave and the analysis is against current load rather than peak load. 
Using previous results, we also observe bounds on channel usage for one of the algorithms. 

\end{abstract}

\keywords{
Reallocation Algorithms,
Windows Scheduling,
Radio Networks,
Unit Fractions Bin Packing
}

\section{Introduction}

We study the Windows Scheduling\footnote{Were it not for the extensive literature on windows scheduling, we would have chosen a different name. After all, the goal of windows scheduling is not to schedule windows but transmissions. Moreover, those transmissions need not be scheduled within fixed epochs, as the term \emph{window} seems to suggest. Nonetheless, we use the current notation for literature consistency.} problem (WS)~\cite{bar2003windows,chan2005temporary,bar2007windows}, 
which is a restricted version of Unit-Fractions Bin Packing (UFBP)~\cite{bar2007windows,balogh2012line,balogh2012semi,han2010dynamic,chan2008dynamic,ivkovic1998fully},
and it is also called Inventory Replenishment in the context of Supply Chain~\cite{yu2013online}.
In brief, the WS problem is to schedule the use of communication channels to clients. Each client $\client_i$ is characterized by an active cycle (with arrival and departure times) and a window $\window_i$ also called laxity\footnote{Throughout the paper, we will use the term \emph{laxity} instead of \emph{window} since it conveys better the concept of maximum delay between transmissions. Nonetheless, we will also use the term \emph{window} when describing the related literature.} (in the context of job scheduling). During the period of time that any given client $\client_i$ is active, there must be at least one transmission from $\client_i$ scheduled in any $\window_i$ consecutive time slots. The optimization criterion is to minimize the number of channels used.
The WS problem appears in many areas such as Operations Research, Networks, Streaming, etc. More application details can be found in~\cite{bar2003windows,chan2005temporary,bar2007windows,yu2013online} and the references therein.

Given that even a restricted version of the WS problem was shown to be NP-hard even for one channel~\cite{NPh}, practical WS solutions are only approximations.
In the WS literature, competitive analysis of the ratio between the number of channels used by an online algorithm with respect to an optimal algorithm has been carried out in two flavors: the maximum ratio for any given time instant (also called \emph{against current load}), and as the ratio of the maxima (also called \emph{against peak load}).
In the model of~\cite{bar2003windows,bar2007windows} both competitive ratios are the same because clients do not leave the system. 
In~\cite{chan2005temporary,yu2013online} the competitive ratio is against peak load. 
In the present work, we carry out competitive analysis against current load. 

In~\cite{bar2003windows,bar2007windows} clients do not leave the system whereas in~\cite{chan2005temporary,yu2013online} clients may leave. 
We further extend the model assuming that clients may be reallocated at some cost. 
As used in~\cite{bender2013reallocation} (for Job Scheduling), we call this class of protocols \defn{Reallocation Algorithms.}  
Reallocation algorithms are a middle ground between online algorithms (where assignments are final) and offline algorithms which can be used repeatedly if reallocations are free.
Reallocation has been studied previously for Load Balancing~\cite{westbrook2000load,sanders2004migration}, and for UFBP in~\cite{balogh2012line,balogh2012semi} where it was called \emph{semi-online algorithms}. 
In~\cite{ivkovic1998fully}, also for UFBP, the reallocation cost paid is only computation time.
Preemptive Job Scheduling has also been studied assuming that a cost for preemption is paid~\cite{heydari2010minimizing,shachnai2005preemption,liu2004minimizing} but preempted jobs in this model are left idling.
In the present work, we assume that the cost of reallocation is a constant amount. That is, we aim to minimize the number of reallocations.

In this paper, we present three protocols for online WS with reallocation. Namely, \preemptive, \lazy, and \constant. In brief, the first is a repeated (upon each client arrival or departure) application of the offline protocol in~\cite{bar2003windows}. Aiming to minimize channel usage, clients are preemptively reallocated to guarantee the same offline packing. Instead, in \lazy, clients are not reallocated as long as a maximum number of channels in use is not exceeded.
The idea is to save reallocations taking advantage of all possible channels. Finally, \constant~is designed to guarantee an amortized constant number of reallocations. The main approach is to classify clients by laxity. 

We evaluate experimentally all three protocols. Our simulations show that, in practice, all three achieve constant amortized reallocations with close to optimal channel usage. Our simulations also expose interesting trade-offs between reallocations and channel usage. 
On the theoretical side, we introduce a new objective function for WS with reallocations, that can be also applied to models where reallocations are not possible. 
We analyze this metric for \constant which, to the best of our knowledge, is the first online WS protocol for dynamic scenarios (clients may leave) with theoretical guarantees (against current load). 
Using previous results, we also observe bounds on channel usage for \preemptive. 

The rest of the paper is organized as follows. First we overview the related work in Section~\ref{relwork}. Section~\ref{model} includes formal definitions of the problem and the model. The details of our contributions are presented in Section~\ref{results}. We present and analyze reallocation protocols in Sections~\ref{preemplazy} and~\ref{constant}, and simulations in Section~\ref{simulations}. 

\section{Related Work}
\label{relwork}

WS is a problem with applications to various areas such as communication networks, supply chain, job scheduling, media on demand systems, etc. We overview here the most related work. 

In~\cite{bar2003windows} the authors studied WS for broadcast systems where pages have to be allocated to slotted broadcast channels.
The allocation comprises a schedule of pages to slots such that the gap between two consecutive slots reserved for page $i$ is at most its window $w_i$.
The model is static in the sense that, once pages arrive to the system, they do not leave. 
Two optimization criteria are studied: to minimize the number of channels used by a fixed set of clients and to maximize the number of clients allocated to a fixed set of channels.
For the former problem they show an offline approximation of $H+O(\ln H)$ where $H=\sum_i 1/w_i$. 
The approach followed is to round down each window size to the largest number of the form $x2^y$, where $x$ is an odd positive integer and $y$ is a non-negative integer. Then, allocate all clients of the same window size to the same channel creating a new one whenever necessary.
For online algorithms, an approximation of $H+O(\sqrt{H})$ channels was shown later in~\cite{bar2007windows}.
The approach is similar but the analysis takes into account that, being online, pages are allocated one by one and decisions are final.

A scenario where pages may leave the system was termed WS with Temporary Items in~\cite{chan2005temporary}.
Bounds in the latter work were proved on the competitive ratio against peak load.
That is, the online and the offline channel-usage maxima compared may occur at different times. That is, taking advantage of (rather than overcoming) windows departures.
Specifically, the authors showed that any online algorithm has a competitive ratio at least $1+1/w_{min}-\epsilon$, for any $\epsilon>0$, independently of knowledge of departure time. On the positive results side, they show an online algorithm that is $2+2/(\lfloor\lfloor w_{min}\rfloor \rfloor-1)$-competitive, for $w_{min}\geq 2$. 
Recently~\cite{yu2013online}, this upper bound was improved to $4w_{min}/(3w_{min}-4)$, for $w_{min}\geq 4$, in the context of supply chain inventory replenishment. 
The algorithm makes use of a combined first fit policy for assignment.

In~\cite{bar2007windows}, WS was studied as a \emph{restricted} version of UFBP. 
The reason is apparent, the combination of windows in WS is not additive as in UFBP.
Hence, UFBP upper bounds do not apply to WS but lower bounds do.
A lower bound of $H+\Omega(\ln H)$ on the online approximation of UFBP, which then applies also to WS, was shown in~\cite{bar2007windows}.
A dynamic version of UFBP where items leave was studied in~\cite{chan2008dynamic}. The authors show an upper bound of $3$ on the competitive ratio of best fit and worst fit allocation policies, and an upper bound of $2.4942$ for first fit. 
UFBP with reallocation was studied in~\cite{balogh2012line,balogh2012semi} as \emph{semi-online} algorithms with \emph{repacking}. 
In their model items do not leave the system, hence only lower bounds apply to WS with reallocation, but only upper bounds are presented. 
In~\cite{ivkovic1998fully}, \emph{fully dynamic} UFBP algorithms apply to settings where items may leave, but the reallocation cost is only bounded in terms of computation time.

\section{Model}
\label{model}

In this Section, we describe the Windows Scheduling problem and introduce 
the notation used throughout the paper.

We assume that time is discrete and it is determined by a global clock that runs in \defn{time steps}. Consider a set of clients $\CL$ that require periodical radio transmissions. Each client $\client_i \in \CL$ is determined by three parameters: \defn{arrival} time $\arrival_i$,  \defn{departure} time $\departure_i$, and \defn{laxity} 
$\window_i$. We say that client $\client_i$ is \defn{active} during the time interval $[\arrival_i,\departure_i]$. 
We assume that there is an infinite number of available radio channels $\CH = \{\channel_j\}_{j\in \mathbb{N}}$. 
During one time step, only one client is able to transmit via a single channel.

A \defn{schedule} for a client $\client_i \in \CL$ is a set $\schedule(\client_i) \triangleq\{(\channel^i_k,t^i_k)\}_{k \in [1,\ell_i]}$,
where, channel $\channel^i_k$ is the reserved channel for $\client_i$'s transmission at time step $t^i_k$.
Time steps $t^i_k$ should satisfy the following conditions:
\begin{eqnarray}
t^i_1 - \arrival_i &\leq& \window_i,\label{l-arrival} \\ 
\departure_i-t^i_{\ell_i} &\leq&\window_i\mbox{ and} \label{l-departure}\\
t^i_{k+1}-t^i_{k} &\leq& \window_i \mbox{ for all } 1 \leq k < \ell_i. \label{l-laxity}
\end{eqnarray}
Parameter $\ell_i$ denotes the number of transmissions performed by client $\client_i$. Note that this parameter is independent for each client, 
it is not set before hand and has to be set by the schedule so that conditions (\ref{l-arrival}), (\ref{l-departure}) and (\ref{l-laxity}) are satisfied.
Thus, conditions (\ref{l-arrival}),  
(\ref{l-departure}) and (\ref{l-laxity})
ensure that each client $\client_i$ transmits only while it is active and that periodical transmissions are so that the laxity condition $\window_i$ for client $\client_i$ is satisfied. 
Since client $\client_i$ has to transmit periodically every (at most) $\window_i$ time steps, we say that the \defn{load of client} $\client_i$ is $1/\window_i$. 

We say that schedule $\schedule(\client_i)$ \defn{reallocates} client $\client_i$ each time that $\channel^i_k \neq \channel^i_{k+1}$ for $1\leq k < \ell_i$. We denote by $reall(\schedule(\client_i))$ the number of times schedule $\schedule(\client_i)$ reallocates client $\client_i$. The total \defn{number of reallocations} of schedule $\schedule(\CL)$ for the set of clients $\CL$ is denoted by $reall(\schedule(\CL))$ and is defined as follows:
$$
reall(\schedule(\CL)) \triangleq \sum_{\client_i \in \CL} reall(\schedule (\client_i)).
$$

We are interested in the \defn{amortized} number of reallocations, that is, the number of reallocations of a schedule divided by the number of arrivals and departures in the system. For that purpose, let us define a \defn{round} as the set of time slots between events (arrival or departure of a client). In other words, at each round occurs exactly one event, either an arrival or a departure of a client. Hence, if $reall_r(\schedule(\CL))$ denotes the number of reallocations up to round $r$ produced by schedule $\schedule(\CL)$, the \defn{amortized} number of reallocations of schedule $\schedule(\CL)$ in the set of clients $\CL$ at round $r$ is
$ reall_r(\schedule(\CL))/r$.

Naturally, there exist a trade-off between the number of reallocations and the number of channels used by a schedule, the more reallocations the less number of channels used at each round. Indeed, if a schedule can reallocate clients freely, it can achieve the  optimal offline number of channels used at each round by simply computing the optimal offline schedule at each round, and reorganizing the schedule via reallocations in order to mimic the optimal offline schedule. Thus, we are interested in understanding this trade-off at each round. For any round $r$, let $\CL(r)\subseteq\CL$ be the set of active clients in the system at round $r$. 
Hence, we define $H_r \triangleq \sum_{\client_i \in \CL(r)} 1/\window_i$ as the \defn{load of the system} at round $r$.  The value $\lceil H_r \rceil$ is a lower bound on the optimal number of used channels at round $r$. Hence, $\CH(\schedule(\CL))_r / \lceil H_r \rceil$  is an upper bound on the competitive ratio of the number of channels used at round $r$, where $\CH(\schedule(\CL))_r $ denotes the number of channels used at round $r$ by schedule $\schedule(\CL)$ . 

Using this notation, we define the online WS problem with reallocations as follows: 
\begin{definition} 
\label{def:metric}
Let $\CL$ be a set of clients revealed online, that is, arrivals and departures of clients are revealed one by one. 
The \defn{Online Windows Scheduling with Reallocations} problem is to determine a schedule $\schedule(\CL)$, possibly with reallocations, 
so that the following sum is minimum.
$$
 \frac{reall_r(\schedule(\CL))}{r} + \frac{\CH(\schedule(\CL))_r}{\lceil H_r \rceil}
$$
\end{definition}

Notice that this metric is against current load and applies to any given round $r$. Should we instead be interested in the maximum competitive ratio, we could simply plug the maximum up to round $r$ instead of the current. Likewise, should the reallocation cost be any known value $c$, it could be simply introduced in the expression above multiplying the first term.

\section{Our Contributions}
\label{results}

The main contributions of this paper follow. It should be noticed that our upper bounds cannot be compared with previous theoretical work on WS because either their models assume that clients do not leave~\cite{bar2003windows,bar2007windows}, or the bounds on channel usage are proved against peak load~\cite{chan2005temporary,yu2013online}, or they apply to the less restrictive UFBP problem~\cite{bar2007windows,chan2008dynamic,balogh2012line,balogh2012semi,ivkovic1998fully}.

\begin{itemize}
\item The presentation of three protocols for Online Windows Scheduling with Reallocation called \preemptive, \lazy, and \constant. \preemptive guarantees low channel-usage because clients are preemptively reallocated to achieve the offline packing of~\cite{bar2003windows}. Attempting to save reallocations while keeping the channel usage bounded, in \lazy clients are not reallocated as long as a maximum number of channels in use is not exceeded.
In \constant, the main approach is to classify clients by laxity, except for ``large'' laxities that are allocated in one special channel to maintain the channel-usage overhead below an additive logarithmic factor. To the best of our knowledge, \constant is the first online WS protocol for dynamic scenarios (clients may leave) with theoretical guarantees (against current load). 

\item The experimental evaluation of all three protocols showing that, in practice, all of them achieve constant amortized reallocations with close to optimal channel usage. Our simulations also expose interesting trade-offs between reallocations and channel usage.

\item The introduction of a new objective function for Online Windows Scheduling with Reallocations, that can be also applied to models where reallocations are not possible. This metric combines linearly the effect of reallocations amortized over rounds with the number of channels used in contrast with the optimal UFBP, which is a lower bound on the WS optimal allocation.

\item Using previous results~\cite{bar2003windows}, for \preemptive we observe an upper bound of $\lceil 2 H_r \rceil$ on channel usage for every round $r$. Given that in \lazy the maximum number of channels is a parameter, its channel usage depends on the implementation. 

\item We prove that, for any round $r$, the number of reallocations required by \constant is at most $3r/2$, and the number of channels used is at most
$OPT(r)+1+\log \left(\min \left\{ w_{\max}(r),\lceil\lceil n(r) \rceil\rceil \right\} / w_{\min}(r) \right)$~\footnote{Throughout, $\log$ means $\log_2$ unless otherwise stated. We define $\lceil\lceil n(r) \rceil \rceil$ as the smallest power of $2$ that is not smaller than $n(r)$.}
, where $OPT(r)$ is the optimal number of channels required, $w_{\max}(r)$ and $w_{\min}(r)$ are respectively the maximum and minimum laxities in the system, and $n(r)$ is the number of clients in the system, all for round $r$. We apply these bounds to our objective function in Definition~\ref{def:metric}.
\end{itemize}


\section{\preemptive and \lazy Algorithms}
\label{preemplazy}
The first two algorithms use the concept of a \emph{broadcast tree} to represent the schedule corresponding to each channel. 

\parhead{Broadcast Tree to Represent a Schedule}
Similarly to the the work by Bar-Noy et al. \cite{bar2003windows} and Chan et al. \cite{chan2005temporary}, we represent a schedule for a channel 
with a tree.
In particular, we use a binary tree where all nodes have exactly zero or two children.  Each leaf of the binary tree has assigned one client (different for each leaf) or the leaf is empty. Each complete binary tree with such an assignment represents uniquely one schedule for a channel. Given a complete binary tree, the schedule picks one leaf at each time step so that either the assigned client transmits or no transmission is produced if the leaf is empty. In order to pick one leaf, walk down from the root up to one leaf following the next recursive rules: the first time a bifurcation is visited go to the left child. The $i$-th time a bifurcation is visited go to the child that was not visited in the previous visit.  Such trees are called \defn{broadcast trees}. 

A used channel will be fully described by its broadcast tree in the following two algorithms. 
Thus, the following two algorithms are determined by \textit{(i)} the procedure to assign clients to a leaf in a broadcast tree when a client arrives and 
\textit{(ii)} the procedure to reallocate clients if required when a client leaves the system.

\parhead{Greedy Construction of a Broadcast Tree}
Consider an empty channel as a single node, the root of its corresponding tree. Define every root as \defn{available}. 
For the first client $\client_1$ that arrives, let $v_1$ be the nonnegative integer such that $2^{v_1} \leq \window_1 < 2^{v_1 +1}$. 
Append to any root a binary tree of height $v_1$ in which for each depth from $1$ to $v_1-1$, there is a single leaf and for depth $v_1$ there are two leaves.
Set every new leaf as available except for one leaf at depth $v$ which is assigned to $\client_1$. 
For the $i$-th client $\client_i$ that arrives, let $v_i$ be the nonnegative integer such that $2^{v_i} \leq \window_i < 2^{v_i +1}$.
If there is an available leaf at depth $v_i$, assign $\client_i$ to that leaf. Otherwise, let $0<u<v_i$ be any value such that there is an open leaf at depth $u$ in some used broadcast tree. 
Append a binary tree of height $v_i - u$ to that leaf in which for each depth from $u$ to $v_1- u -1$, there is a single leaf and for depth $v_1$ there are two leaves.
Set every new leaf as available except for one leaf at depth $v_i$ which is assigned to $\client_i$.
If  no such $u$ exists, append the described tree to an available root. In this case consider $u=0$. 

\parhead{\preemptive}
The \defn{\preemptive} procedure maintains the following invariant: for each depth there is at most one broadcast tree with a leaf available. 
If the invariant is violated after some departure, it means that there are two broadcast trees with an available leaf in the same depth. 
Then, the branches that hang from the twin node of the available leaf can be hanged from the same tree to reinstate the invariant. 
The \preemptive procedure does so minimizing the total number of reallocations, that is, moving the branch that hangs from the broadcast tree with less clients assigned.

\parhead{\lazy}
The \defn{\lazy} procedure reallocates only when the fraction $ \CH(\schedule(\CL))_r/\lceil H_r \rceil$ exceeds a threshold $T$ after one departure.
It exhaustively reallocates clients until no more reallocations can be made. Reallocations are done according to the invariant used by the \preemptive procedure merging the smallest depth with two available leaves at the same depth.

The \defn{\preemptive} and \defn{\lazy} algorithms use the greedy construction of broadcast trees when clients arrive and preemptive and lazy reallocations procedures when clients leave the system, respectively. 
Lemma $5$ of Bar-Noy et al. \cite{bar2003windows} implies that the \preemptive algorithm guarantees $ \CH(\schedule(\CL))_r < \lceil 2 H_r \rceil$ for every round $r$. Hence, $ \CH(\schedule(\CL))_r/\lceil H_r \rceil < 2$ for every round $r$ when $\schedule(\CL)$ is constructed according to the \preemptive algorithm. On the other hand, by definition of the \lazy procedure, the \lazy algorithm guarantees that  $ \CH(\schedule(\CL))_r/\lceil H_r \rceil \leq T$ when  $\schedule(\CL)$ is constructed according to the \lazy algorithm. In Section \ref{simulations}, we study via experiments the behavior of $ reall_r(\schedule(\CL))/r$ for the \preemptive and \lazy algorithms. In the particular implementation of the \lazy algorithm included in this work, we set the threshold $T$ equal to $4\sqrt{H_r}$. Hence, in that case it holds $\CH(\schedule(\CL))_r/\lceil H_r \rceil \leq 4\sqrt{H_r}$.  For both algorithms, \preemptive and \lazy, we show experimentally that $ reall_r(\schedule(\CL))/r \leq 1$.  

\section{\constant}
\label{constant}

In this section, we present a reallocation algorithm that guarantees $O(1)$ reallocations amortized on rounds.
The details of the protocol can be found in Algorithm~\ref{alg:O(1)realloc}. 
Bounds on channel usage and reallocations are proved in Theorem~\ref{thm:O(1)realloc}. 
Corollary~\ref{corollary} establishes these bounds in our objective function. 
For convenience, for any number $x$, we define the \defn{hyperceiling} of $x$, denoted as $\lceil\lceil x \rceil \rceil$, to be the smallest power of $2$ that is not smaller than $x$.
For this algorithm, we restrict the input to laxities that are powers of $2$. The study of inputs with arbitrary laxities is left for future work.

The intuition of the protocol is the following. 
When a new client $\client_i$ arrives and there are already $n-1$ clients in the system, for $n\geq1$, we distinguish two cases.
If the laxity of $\client_i$ is at least $2\lceil\lceil n \rceil\rceil$, assign $\client_i$ to a special channel called \defn{big channel}. All clients allocated to the big channel transmit with period $\lceil\lceil n \rceil\rceil$ so, because there are $n$ clients in the system, one big channel is enough. All the other channels being used are called \defn{small channels}.
Otherwise, if the laxity is $w_i<2\lceil\lceil n \rceil\rceil$, assign client $\client_i$ to a channel reserved for laxities $w_i$, we call it \defn{$w_i$-channel}. If such channel does not exist or all $w_i$-channels are full, reserve a new one.
For any laxity $w_i$, all clients allocated to a $w_i$-channel transmit with period $w_i$. That is, a maximum of $w_i$ clients can be allocated to a $w_i$-channel.
When a client $\client_j$ of laxity $w_j$ leaves a channel $C$, if $C$ is the big channel do nothing.
Otherwise, reallocate a client from the $w_j$-channel of minimum load (if any other) to the slot left by $\client_j$.

With each arrival or departure the number of clients $n$ change.
If, upon an arrival, $\lceil\lceil n\rceil\rceil$ becomes larger than the laxity of some clients allocated to a big channel, reallocate those clients to other channels according to laxity, reserving new channels if necessary.
Because $n$ was doubled since the allocation of these clients, these reallocations are amortized by the arrivals that doubled $n$.
If, upon a departure, $2\lceil\lceil n\rceil\rceil$ becomes smaller than the laxity of some clients, reallocate those clients to a big channel, releasing the reservation of the channels that become empty.
Because $n$ was halved since the allocation of these clients, these reallocations are amortized by the departures that halved $n$.


\begin{algorithm}[tbp]
\label{alg:O(1)realloc}
\caption{$O(1)$ reallocations. $\lceil\lceil n\rceil\rceil$ is the largest power of $2$ that is not greater than $n$. For any laxity $w$, all clients allocated to a $w$-channel are scheduled to transmit with period $w$.}
\small
\SetKwFor{Upon}{upon}{do}{endupon}
\SetKwFor{Task}{Task}{}{endtask}
\DontPrintSemicolon
   $n\leftarrow 0$\label{init}\tcp*[f]{active clients count}\;
   $\tau\leftarrow 2$\label{init}\tcp*[f]{big channel threshold}\;
   \textbf{start} tasks $1$ and $2$\;
\Task{1}{
\Upon{arrival of client $\client_i$}{
   $n\leftarrow n+1$\label{arrival}\;
   \If(\tcp*[f]{consolidate the big channel}){$2\lceil\lceil n \rceil\rceil > \tau $}{ \label{line:bigoutbegin}
      $\tau\leftarrow2\lceil\lceil n \rceil\rceil$\;
         \ForEach{client $\client_j$ in the big channel such that $w_j<\tau/2$}{\label{line:bigoutmid}
            \lIf{all $w_j$-channels are full}{reserve a new $w_j$-channel}
            reallocate $\client_j$ to the $w_j$-channel of minimum load\;\label{line:bigoutend}\;
         }
         \ForEach{client $\client_j$ in the big channel}{
            re-schedule $\client_j$ to transmit with period $\tau/2$\;
         }
   }
   \If(\tcp*[f]{allocate new client}){$w_i\geq \tau$}{ \label{line:biguponarrival}
      allocate $\client_i$ to the big channel to transmit with period $\tau/2$\;
   }\Else{\label{line:smalluponarrival}
      \lIf{all $w_i$-channels are full}{reserve a new $w_i$-channel}
      allocate $\client_i$ to the $w_i$-channel of minimum load\; 
   }
}
}
\Task{2}{
\Upon{departure of client $\client_i$ from channel $c$}{
   $n\leftarrow n-1$\label{departure}\;
   \If(\tcp*[f]{consolidate $w_i$-channels}){$c$ is not the big channel}{
      \lIf{$c$ is empty}{release $c$}
      \ElseIf{there is a $w_i$-channel $c'\neq c$ that is not full}{ \label{line:simplebegin}
            reallocate a client from $c'$ to $c$\;\label{line:simpleend}
      }
   }
   \If(\tcp*[f]{consolidate the big channel}){$2\lceil\lceil n \rceil\rceil < \tau$}{ 
      $\tau\leftarrow2\lceil\lceil n \rceil\rceil$\;
         \ForEach{client $\client_j$ in the big channel}{
            re-schedule $\client_j$ to transmit with period $\tau/2$\;
         }
      \ForEach{$w$-channel $c''$ such that $w>2\tau$}{\label{line:biguponconsolidation}
            \ForEach{client $k$ in $c''$}{
               reallocate $k$ to the big channel to transmit with period $\tau/2$\;\label{line:biguponconsolidationmid}
            }
            release $c''$\;\label{line:biguponconsolidationend}
      }
   }
}
}
\end{algorithm}


The following theorem bounds the number of reallocations and the number of channels used by the \constant algorithm.

\begin{theorem}
\label{thm:O(1)realloc}
Given a set of clients $\CL$, the schedule $S(\CL)$ obtained by the \constant algorithm requires at most $3r/2$ reallocations up to round $r$. Additionally, for any round $r$ such that $\CL(r)\neq\emptyset$, the number of reserved channels is at most 
\begin{align*}
OPT(r)+1+\log \frac{\min\left\{\max_{i\in\CL(r)}\{w_i,\lceil\lceil|\CL(r)|\rceil\rceil\}\right\}}{\min_{i\in\CL(r)}w_i},
\end{align*}
where $OPT(r)$ is the minimum number of channels required to allocate the clients in $\CL(r)$.
\end{theorem}


\begin{proof}
The bound on the number of channels follows from the algorithm. Specifically, for any round $r$, there are at most $OPT(r)$ $w_i$-channels full, and there is at most one big channel. With respect to the not-full $w_i$-channels, the maximum $w_i$ is either $\max_{i\in\CL(r)}w_i$ or $\lceil\lceil|\CL(r)|\rceil\rceil$, whatever is smaller. Given that all laxities are powers of $2$ the bound follows.

To bound the reallocations, we map reallocations to arrivals or departures. Given that rounds are defined by arrivals and departures, the amortized bound follows. The mapping is the following. Clients are reallocated due to one of three possible events as follows. 

\begin{enumerate}
\item\label{event1}
Some client $\client_j$ with laxity $w_j$ departed from a $w_j$-channel $c$ that was full. Then, if there is some other $w_j$-channel $c'$ that is not full, some client $\client_i$ allocated to $c'$ is reallocated to the slot left by $\client_j$ in $c$ (see Lines~\ref{line:simplebegin}-\ref{line:simpleend} in Algorithm~\ref{alg:O(1)realloc}). 
\item\label{event2}
Upon the arrival of some client, the total number of clients $n$ increases making $2\lceil\lceil n \rceil\rceil$ larger than the big-channel threshold. Then, any client $\client_i$ allocated to the big channel whose laxity is $w_i<\lceil\lceil n \rceil\rceil$ is reallocated to a $w_i$-channel (see Lines~\ref{line:bigoutbegin}-\ref{line:bigoutend} in Algorithm~\ref{alg:O(1)realloc}). 
\item\label{event3}
Upon the departure of some client, the total number of clients $n$ decreases making $2\lceil\lceil n \rceil\rceil$ smaller than the big-channel threshold. Then, all clients in all $w$-channels such that $w>4\lceil\lceil n \rceil\rceil$ are reallocated to the big channel (see Lines~\ref{line:biguponconsolidation}-\ref{line:biguponconsolidationend} in Algorithm~\ref{alg:O(1)realloc}). 
\end{enumerate}
No other event triggers a reallocation. 

Now, we define the mapping. For Event~\ref{event1}, the reallocation of $\client_i$ is mapped to the departure of $\client_j$. To define the mapping for Event~\ref{event2}, we need the following lemmas.

\begin{lemma}
\label{claim:doubled}
Consider a client $\client_i$ that has to be reallocated from the big channel in Lines~\ref{line:bigoutmid}-\ref{line:bigoutend} of Algorithm~\ref{alg:O(1)realloc}. 
Let $n$ be the number of clients in the system at the time of reallocation in Line~\ref{line:bigoutend}, and
$n'$ be the number of clients in the system at the time of the last allocation of $\client_i$.
Then, it is $n\geq 2n'$.
That is, after the last allocation of $\client_i$ to the big channel, the total number of clients in the system at least has doubled.
\end{lemma}
\begin{proof}
In the following, all line numbers refer to Algorithm~\ref{alg:O(1)realloc}.
Clients are reallocated from the big channel because their laxities are strictly smaller than $\tau/2= \lceil\lceil n \rceil\rceil$ (see Line~\ref{line:bigoutmid}).
In order to be allocated to the big channel, $\client_i$ must have a laxity at least $2\lceil\lceil n' \rceil\rceil$ if it was upon arrival (see Line~\ref{line:biguponarrival}), or at least $4\lceil\lceil n' \rceil\rceil$ if it was upon consolidation of the big channel (see Line~\ref{line:biguponconsolidation}). In either case, it must be $2\lceil\lceil n' \rceil\rceil\leq w_i < \lceil\lceil n \rceil\rceil$, which implies that $n\geq 2n'$.
\qed\end{proof}

\begin{lemma}
\label{claim:bigout}
In any given round, at most half of the clients in the system are reallocated from the big channel after executing Lines~\ref{line:bigoutmid}-\ref{line:bigoutend} of Algorithm~\ref{alg:O(1)realloc}.
\end{lemma}
\begin{proof}
Out of the set of clients being reallocated from the big channel at a given round $r$, consider the client $\client_i$ that was allocated last, say, in some round $r'<r$.
From Lemma~\ref{claim:doubled}, we know that between $r'$ and $r$ the number of clients in the system has at least doubled. Furthermore, we know that none of the clients that arrived after $r'$ has to be reallocated in round $r$, because $\client_i$ was the last one. Hence, it cannot be that more than half of the clients in the system are reallocated from the big channel in round $r$.
\qed\end{proof}

Now we define the mapping for Event~\ref{event2}. Let $n$ be the number of clients in the system at the time of Event~\ref{event2}. As shown in Lemma~\ref{claim:bigout}, \emph{at most} $n/2$ clients have to be reallocated. And, as shown in Lemma~\ref{claim:doubled}, after the last allocation to the big channel of any client that has to be reallocated, the total number of clients in the system \emph{at least} has doubled. Hence, the at most $n/2$ reallocations are mapped to the at least $n/2$ arrivals.

To define the mapping for Event~\ref{event3}, we need the following lemma.

\begin{lemma}
\label{claim:halved}
Consider a client $\client_i$ that has to be reallocated to the big channel in Lines~\ref{line:biguponconsolidation}-\ref{line:biguponconsolidationend} of Algorithm~\ref{alg:O(1)realloc}. 
Let $n$ be the number of clients in the system at the time of reallocation in Line~\ref{line:biguponconsolidationmid}, and
$n'$ be the number of clients in the system at the time of the last allocation of $\client_i$.
Then, it is $n\leq n'/2$.
That is, after the last allocation of $\client_i$, the total number of clients in the system at least has halved.
\end{lemma}
\begin{proof}
In the following, all line numbers refer to Algorithm~\ref{alg:O(1)realloc}.
Clients are reallocated to the big channel because their laxities are strictly larger than $2\tau = 4\lceil\lceil n \rceil\rceil$ (see Line~\ref{line:biguponconsolidation}).
Because $\client_i$ was not in the big channel, $\client_i$ must have a laxity strictly smaller than $2\lceil\lceil n' \rceil\rceil$ if the last allocation was upon arrival (see Line~\ref{line:smalluponarrival}), or strictly smaller than $\lceil\lceil n' \rceil\rceil$ if it was upon consolidation of the big channel (see Line~\ref{line:bigoutmid}). 
In either case, it must be $4\lceil\lceil n \rceil\rceil< w_i < \lceil\lceil n' \rceil\rceil$, which implies that $n\leq n'/4\leq n'/2$.
\qed\end{proof}

Finally, the mapping for Event~\ref{event3} is the following. Let $n$ be the number of clients in the system at the time of Event~\ref{event3}. As shown in Lemma~\ref{claim:halved}, after the last allocation of any client that has to be reallocated, the total number of clients in the system at least has halved. Hence, the at most $n$ reallocations are mapped to the at least $n$ departures.

In the mapping above, there is at most one reallocation for each arrival and at most two reallocations for each departure. Given that there cannot be more departures than arrivals, the number of reallocations up to round $r$ are at most $r/2+r$. Hence, the claim follows.
\qed\end{proof}


The following corollary is a direct consequence of Theorem~\ref{thm:O(1)realloc} and the fact that $OPT(r)\geq\lceil\sum_{i\in \CL(r)} 1/w_i\rceil=\lceil H_r \rceil$.

\begin{corollary}
\label{corollary}
Given a set of clients $\CL$, the schedule $S(\CL)$ obtained by the \constant algorithm achieves, for any round $r$, is
$ reall_r(\schedule(\CL))/r + \CH(\schedule(\CL))_r/\lceil H_r \rceil \leq 
5/2 + (1+\log(\min\{\max_{i\in\CL(r)}w_i,\lceil\lceil |\CL(r)|\rceil\rceil\}/\min_{i\in\CL(r)}w_i))/\lceil H_r \rceil$.
\end{corollary}


%

\section{Simulations}
\label{simulations}

\parhead{Model}
The deployment of the proposed algorithms on a real environment will require involvement of a large number of active users and resources, which is very hard to coordinate and build, and would prevent repeatability of results. Thus, simulation appears to be the easiest way to analyze the different proposed online reallocation strategies. Based on the simulation results, we can later encourage or discourage the deployment on a real production environment. Next, we present the simulation model implemented for evaluating the performance of the previously proposed reallocation policies, named \emph{CommunicationChannelsSim} (CCSim). The simulated communication channels have been performed by means of \texttt{SimJava 2.0} \cite{SimJava:www}. \texttt{SimJava} is a discrete event, process oriented, simulation package. It is an API that augments Java with building blocks for defining and running simulations.

All the simulations are performed from the point of view of one user. This user submits \emph{clients} to the scheduler. The clients are loaded from an XML input file. For our experiments we have created three different input files each containing 4000 clients with different laxities. These scenarios represent common situations on which clients with different laxities arrive to the system and departure in an instant of time that depends directly on the laxity demanded by the client. The better the laxity, the sooner he will leave the system. In the same way, the scheduler performs one of the proposed reallocation strategies and acts accordingly when a client arrives or leaves the system. We have implemented three versions of CCSim that only differ in the reallocation policy implemented. As a result, we isolate the reallocation strategy as the only factor that can cause number of channels and number of reallocations variations between these three CCSim versions. 

As we mentioned before, we have created our own input files. We have defined the structure of the files using the XML language. Thus, each input file contains 4000 clients with the following characteristics:
\begin{dinglist}{93}
    \item \texttt{id}: each client has its own identifier to differentiate it from the rest.
    \item \texttt{t\_arrive}: arrival simulation time of the client in seconds, for example, 899 seconds. 
    \item \texttt{size}: laxity of the client as a real value, for example, 6.628461669685978. The laxity can be calculated using a Gaussian (normally) distributed double value or it can be chosen uniformly. Thus, we can generate input files following Normal, Uniform or mixed laxity distributions.
    \item \texttt{w\_size}: the laxity is later rounded down to the larger power of 2, for our example, its value would be 4. 
    \item \texttt{t\_leave}: simulation time instant in seconds at which the client leaves the system. This value is calculated based on the laxity. Thus, if the client has a laxity less than or equal to 30, \texttt{t\_leave} will be \texttt{t\_arrive} plus a random number generated uniformly between 500 and 1000. Conversely, if the client has a laxity greater than 30, \texttt{t\_leave} will be \texttt{t\_arrive}  plus a random number generated uniformly between 1000 and 1500. As a consequence, clients with a smaller laxity will remain less time in the system. For instance, the client of the example will leave the system at second 1737.0.
\end{dinglist} 

Explanations for the main CCSim participating entities and its simulation setup follow.

\textbf{CCSim} represents the complete simulation, and is responsible for the creation of the main simulated entities: Scheduler and User. When the simulation starts, CCSim creates 1 User and 1 Scheduler.

The \textbf{Scheduler} entity represents a generic scheduler implementing the corresponding reallocation algorithm. When a Client \emph{arrives}, the Scheduler allocates it in a Broadcast Tree Node. This operation may involve creating a new Tree. Also, when a Client leaves, the Scheduler removes it from the corresponding Tree Node. This operation may imply reallocating Clients and deleting Trees.

The \textbf{User} models a user that submits Clients to the Scheduler. The User is responsible of sending the only two system events: Client arrival and Client departure by using Client's \texttt{t\_arrive} and \texttt{t\_leave} times. 

The \textbf{Client} entity represents a generic Client submitted to the Scheduler. This entity provides specific information about each client as defined in the XML input file. 

The Broadcast \textbf{Tree} entity, as described in Section \ref{preemplazy}, represents a channel with up to five different levels representing five Client laxities (2, 4, 8, 16, 32). Thus, at each level there are at most $2^{i}$ Nodes, with $i$ being a number between 1 and 5. Clients are allocated in their corresponding Node level according to their rounded down laxity \texttt{w\_size}. Each tree is characterized by a root Node.

The \textbf{Node} entity represents a Tree Node in which a Client can be allocated. Each node knows its level or laxity, if it is a root node, a left node or a right node. Also, each node points to its left subtree and to its right subtree.

\parhead{Discussion}
Our simulations show similar behaviors for normal and uniform input distributions.
All three protocols achieve constant amortized reallocations with close to optimal channel usage. 
Figures~\ref{fig1} and~\ref{fig2} illustrate the performance of the algorithms along rounds for a normally distributed input.
Figure~\ref{fig:HvsR} illustrates the trade-offs between channel usage and reallocations for uniform and normally distributed laxities.

We can observe in Figure~\ref{fig:load} the overall load of the system for a representative input. 
The load was increased over 1000 rounds until reaching a peak and decreased during the following 1000 rounds until reaching low load. Afterwards, the increase/decrease procedure was repeated but now maintaining the system loaded for approximately 4000 rounds.
We can see in Figure~\ref{fig:realloc} that \preemptive incurs in more reallocations than \lazy and \constant, but still for all three the amortized reallocations are below $1$. 
When the system is loaded, we observe in Figure~\ref{fig:ratio} that the competitive ratio stays around $1.5$ for all algorithms. 
When the system has low load, \preemptive still maintains the competitive ratio below $2$, whereas \lazy increases up to the parametric maximum channel usage. For these simulations, that maximum was $H_r+4\sqrt{H_r}$, that is, the online maximum channel usage of~\cite{bar2007windows} for clients that do not leave. With respect to \constant with low load, the competitive ratio is higher because, for simplicity, the simulation was carried out without using a big channel (refer to the \constant algorithm). 
\preemptive and \lazy reflect clearly the trade-off between reallocations and channel usage. While the former incurs in more reallocations than the latter, the reverse is true with respect to channel usage. This trade-off becomes more dramatic if the maximum number of channels allowed in \lazy is increased.
\constant, on the other hand, which out of the three is the only algorithm providing theoretical guarantees, has low amortized reallocations (less than $0.5$) while maintaining a low competitive ratio when the system is loaded. For systems with frequent low load, the big channel should be implemented.
Figure~\ref{fig:objective} illustrates the combination of these factors in our objective function.

\begin{figure}[htb]
\begin{center}
\subfigure[][Current load as $H_r$.]{
\includegraphics[width=.9\textwidth]{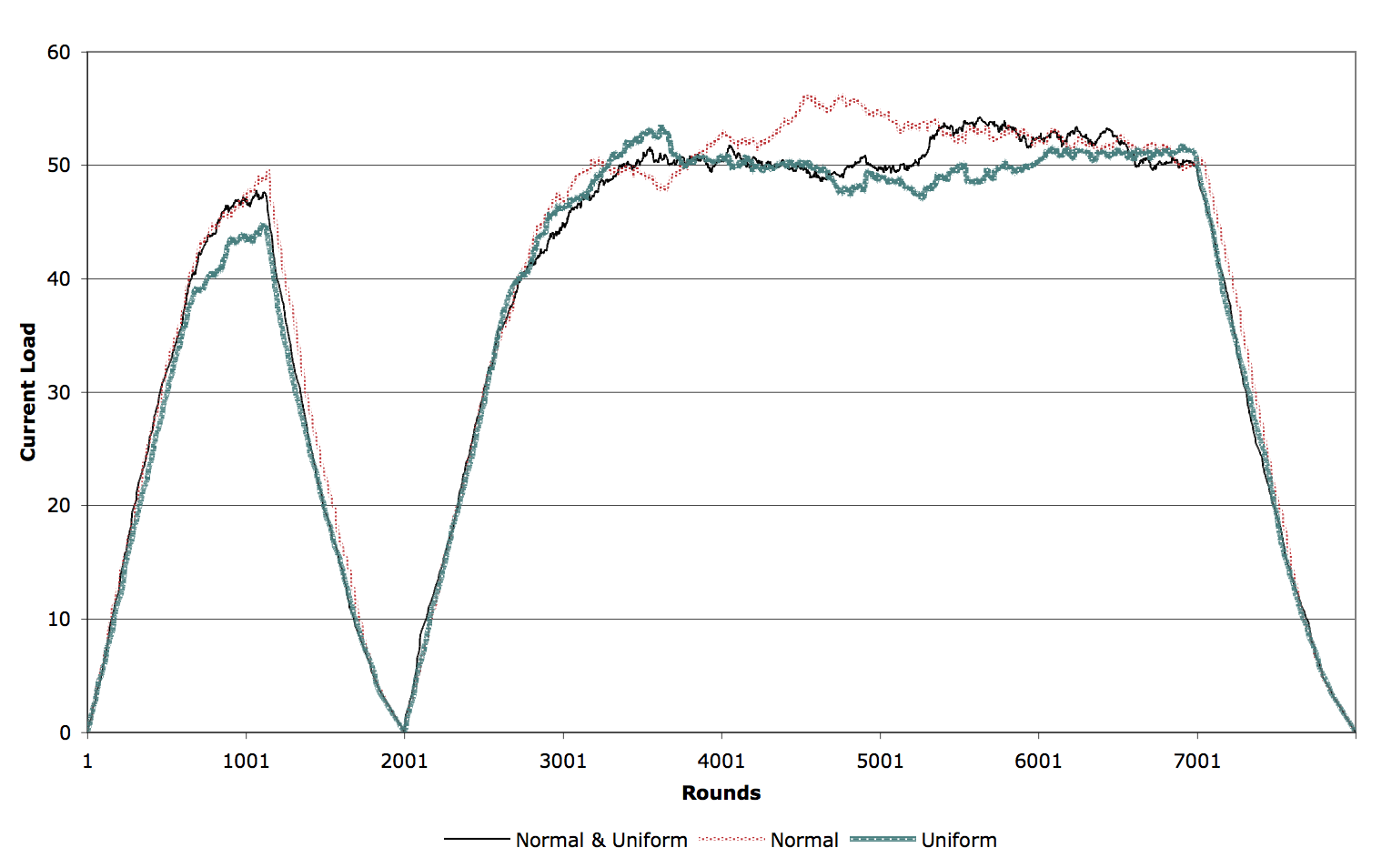}
\label{fig:load}
}
\subfigure[][Objective function.]{
\includegraphics[width=.9\textwidth]{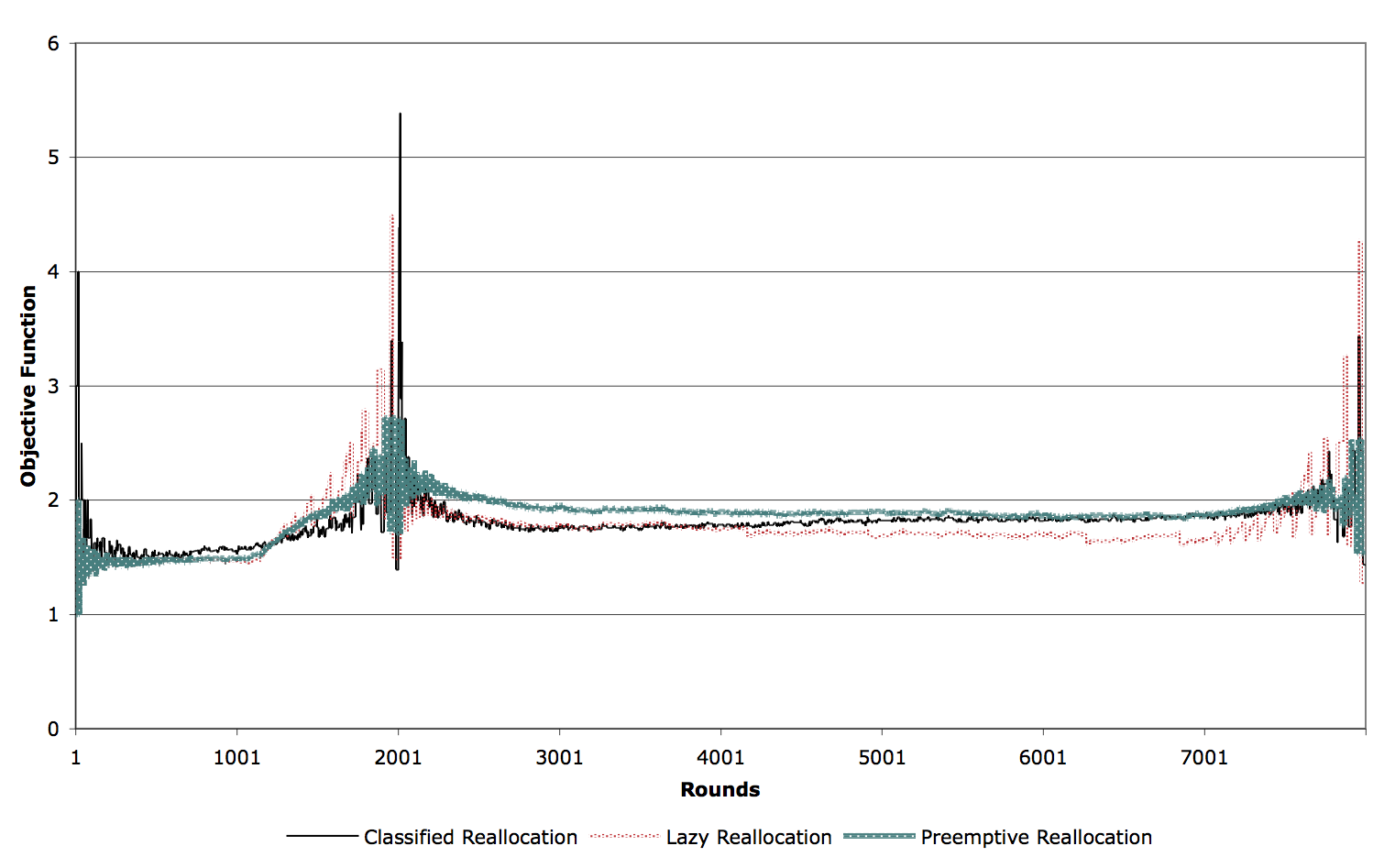}
\label{fig:objective}
}
\caption{Performance along rounds for normally distributed inputs.}
\label{fig1}
\end{center}
\end{figure}
\begin{figure}[htb]
\begin{center}
\subfigure[][Channel usage competitive ratio.]{
\includegraphics[width=.9\textwidth]{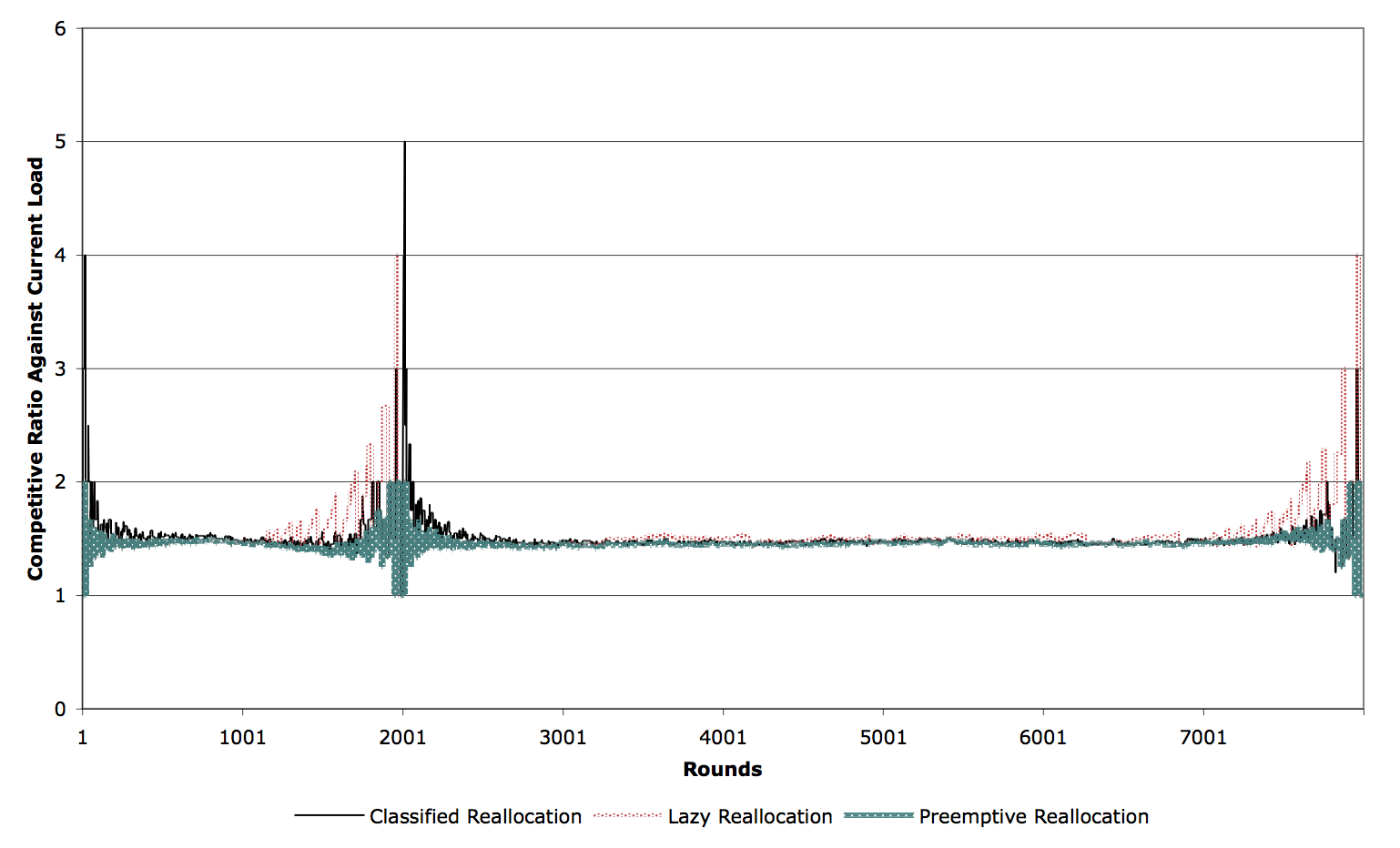}
\label{fig:ratio}
}
\subfigure[][Amortized reallocations.]{
\includegraphics[width=.9\textwidth]{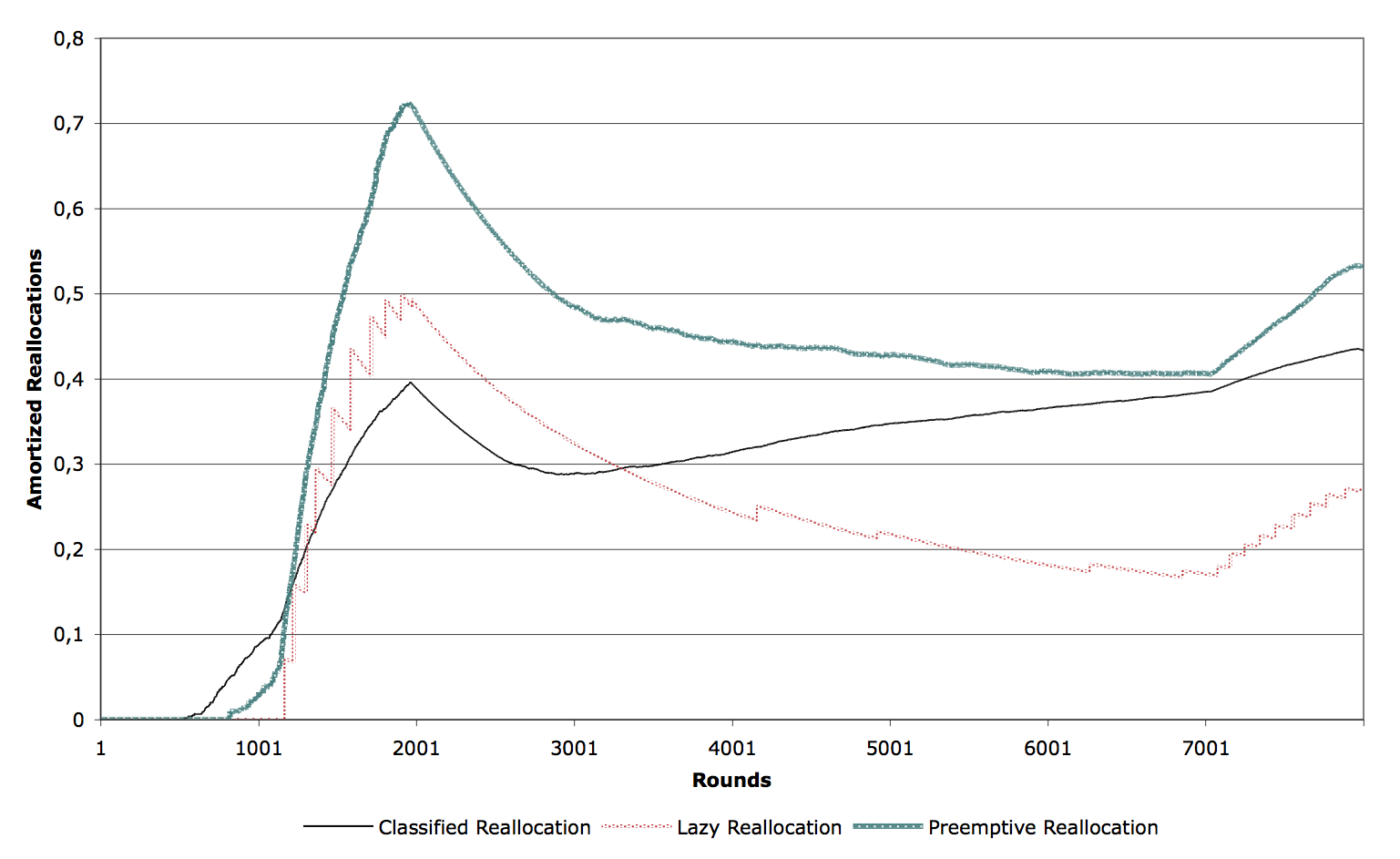}
\label{fig:realloc}
}
\caption{Performance along rounds for normally distributed inputs.}
\label{fig2}
\end{center}
\end{figure}

\begin{figure}[htb]
\begin{center}
\subfigure[][Uniform distribution of laxities]{
\includegraphics[width=.9\textwidth]{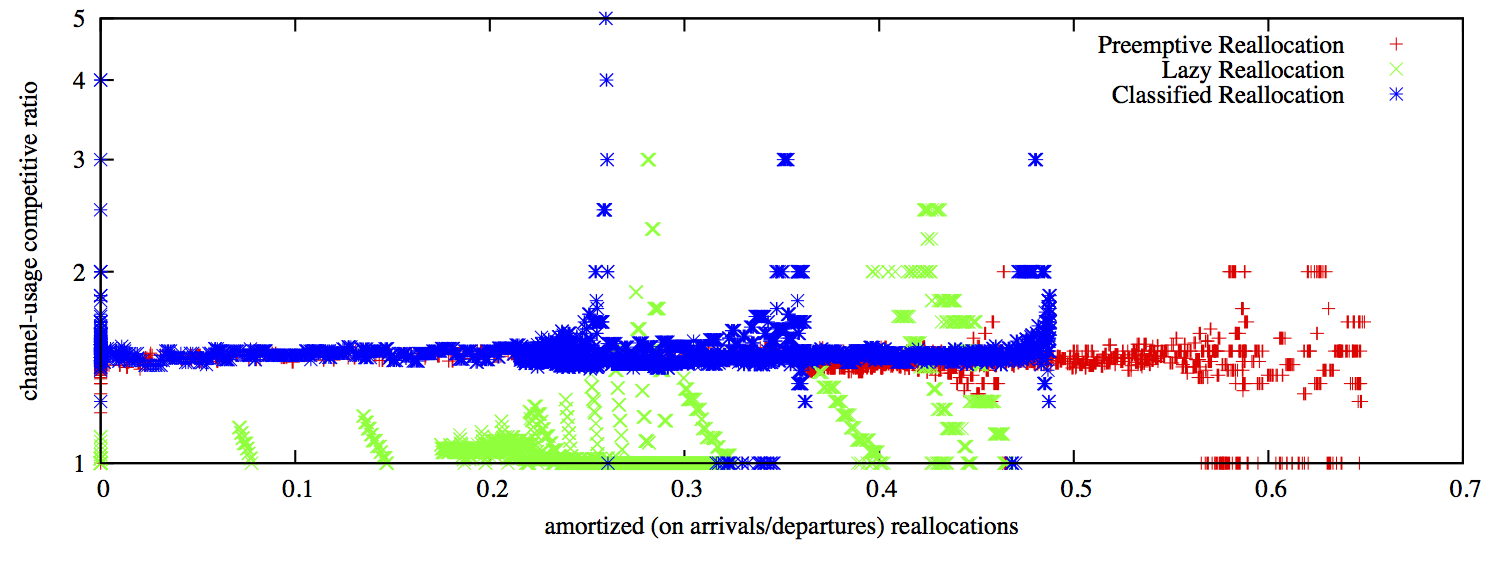}
}
\subfigure[][Normal distribution of laxities]{
\includegraphics[width=.9\textwidth]{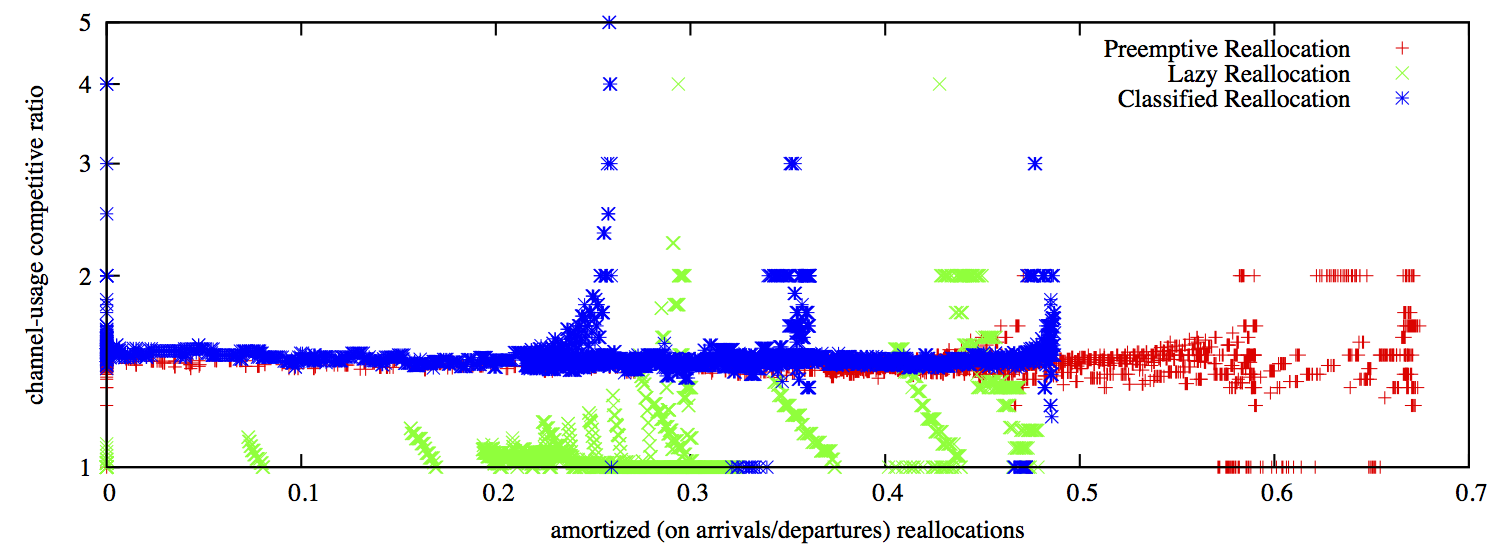}
}
\caption{Channel usage vs. reallocations}
\label{fig:HvsR}
\end{center}
\end{figure}

\bibliographystyle{plain}
\bibliography{FCLMT_WS_arxiv}

\end{document}